\documentclass{article}

\usepackage[english]{babel}

\usepackage[letterpaper,top=2cm,bottom=2cm,left=3cm,right=3cm,marginparwidth=1.75cm]{geometry}

\usepackage{amsmath}
\usepackage{graphicx}
\usepackage[colorlinks=true, allcolors=blue]{hyperref}
\usepackage{ntheorem}
\usepackage{amsmath}
\usepackage{amsfonts}

\usepackage{algorithm}
\usepackage{algorithmic}
\floatname{algorithm}{Procedure}

\numberwithin{figure}{section}

\newtheorem{remark}{Remark}
\newtheorem{definition}{Definition}
\newtheorem*{proof}{Proof}
\newtheorem{proposition}{Proposition}
\newtheorem{lemma}{Lemma}
\newtheorem{theorem}{Theorem}

\newtheorem{corollary}{Corollary}
\begin{document}

\title{PINN-based viscosity solution of HJB equation}
\author{Tianyu Liu, Steven Ding, Jiarui Zhang, Liutao Zhou}
\date{}
\maketitle

\begin{abstract}
This paper proposed a novel PINN-based viscosity solution for Hamiltonian-Jacobi-Bellman (HJB) equations corresponding for both infinite time and limited time optimal control cases. Although there exists PINN-based work for solving HJB, to the best of our knowledge, none of them gives the solution in viscosity sense. This paper reveals the fact that using the (partial) convex neural network, one can guarantee the viscosity solution. Also we proposed a new training method called strip training method, through which the neural network (NN) could easily converge to the true viscosity solution of HJB despite of the initial parameters of NN, and local minimum could be effectively avoided during training.                                            
\end{abstract}
\section{Introduction}

Physics-Informed Neural Network(PINN) is a powerful tool for solving Partial Differential Equation(PDE) and even for identifying the missing parameters. Since the  conception of PINN was first proposed by M. Raissi etc. in \cite{raissi2019physics}, lots of PINN-based applications have been created and showed us its great power in solving PDEs that come from various practical problems, e.g. fluid mechanics, medical diagnosis, etc.\cite{sahli2020physics,cai2021physics}.  Among all the problems, the one that draws great intension for us is the solving of the Hamiltonian-Jacobi-Equation (HJ equation). It is worth noticing that HJ equation has a deep connection with the Regular equations in classical mechanics. On the one hand, the existence of the HJ equation comes from the transformation of Hamilton quantity in order to solve the regular equations in Hamiltonian system, on the other hand, the regular equations are the characteristic lines of the HJ equation. 

Another important application of the HJ equation originates from optimal control, where the Hamiltonian has the form of minimum of a family of linear functions. In this case the HJ equation has the name of HJB where the B represents the Bellman in honor of his contribution to principle of optimality. With the HJB equation, it is easy to find the optimal cost for a finite or infinite horizon optimal control problem. Actually, the difference between the HJB and the popular reinforcement learning technique lies in the knowledge of system or environment model, and also if we know the initial condition when start learning. With an accurate model of the system, we can theoretically solve the minimum cost and the corresponding optimal control very easily.

But there is a huge gap between reality and theory: the difficulty of solving the HJB by using traditional numerical methods such as Finite Difference (FD) method and Semi-Lagrangian(SL) method will increase exponentially with the increase of the state dimension of the system. This problem is called the "curse of dimensionality"\cite{rust1997using}. Since PINN doesn't need to discretize the differential (FD method) or search step by step through discrete dynamic programming (SL method), PINN has theoretically a better performance especially when come across high dimensional system. Some of the existing work using PINN for solving HJB equations could be seen in \cite{furfaro2022physics,mukherjee2023bridging,he2023learning}

Among all these research, a problem is that none of those research considered the solution in viscosity sense. Actually there exists more than one locally Lipschitz continuous function that satisfy the HJB almost everywhere. The key point is to find the viscosity solution that corresponds to the optimal control. Otherwise, if the PINN is trained directly with the loss term from the HJB, the neural network will converge to other solutions instead the viscosity one, unless the start point of training happens to be close to the answer. The existing solution is to first train the NN to a "guess" solution, which is usually based on the quadratic form of system states. However, it is obvious that this guess method can not guarantee to be correct every time.

In this paper, two PINN-based viscosity solutions for HJB equation are proposed. One is adding an extra term in the loss function, the other is using directly a convex network. Section \ref{sec:motivation} and Section \ref{sec:viscosity introduce} gives the motivation of this paper and some basic introduction of the viscosity solution. Section \ref{sec:proof of viscostiy} proves the equivalence of the unique solution of HJB and the convex property of the solution. Section \ref{sec: numerical method} provides two methods to guarantee the uniqueness of solution from neural network and the constructive proof of the convex PINN. Section \ref{sec:simulation} gives the simulation results of the both methods.

\section{Motivation}
\label{sec:motivation}
Before formally introducing the scheme of the viscosity solution of the HJB equation,
it is of great importance to introduce the motivation of studying
the solution in viscosity sense instead of classical sense. On the 
one hand, it is too strong a condition to require differentiable at 
every point in the domain of a function. Actually, the boundary 
value problem for a first order P.D.E. does not admit a global $C^1 
$ solution in general. Hence it is interesting for us to study the 
solution in more generalised case, differentiable except limited 
points. On the other hand, there could be more than one solution 
even if we could find a smooth solution for our problem. An example 
is given below.

\begin{align*}
		\dot{x} &= ax+bu \\
		V^{*} &= \min_{u \in U} \int_{0}^{\infty} \left( x^2 +u^2\right) dt
\end{align*}
where $x$, $u$ are the system states and control inputs in $\mathbb{R}^{1}$. The HJB equation corresponding is given below
\begin{equation*}
	H=\min_{u \in U}\left\{ V_x\dot{x}+x^2 +u^2 \right\}=0
\end{equation*} 

with the boundary condition $V(0)=0$. H could be rewritten as 
\begin{equation*}
	H=-\frac{1}{4} {V_x}^2b^2 + axV_x + x^2=0
\end{equation*} 

From the above equation, it could be imagined that there are at lease two analytical solutions. Actually these solutions are

\begin{align*}
		V_x &=2\frac{a \pm \sqrt{a^2+b^2}}{b^2}x  \\
		V(x) &= \frac{a \pm \sqrt{a^2+b^2}}{b^2} x^2
\end{align*}
we could see that both solutions satisfy the boundary condition, which leads to a natural question that which is the true solution corresponding to our optimal control? It is noticeable that in these two solutions, one is convex, and the other is concave. In our example, it is easy to prove that, the convex one
\begin{align*}
		V(x) &= \frac{a + \sqrt{a^2+b^2}}{b^2} x^2
\end{align*}
is our true solution. 
 
For a more complex multi-dimensional and nonlinear problem, the solution of HJB could be more complicated. Actually for a nonlinear system, the generalized solution, i.e. locally Lipschitz continuous solution that satisfies HJB almost everywhere, doesn't have uniqueness even the stability\cite{bardi1997optimal}. Next two sections would show that the optimal control problem under the context of concave Hamiltonian would lead to the unique semiconvex solution of equation, just like the example shows.
\section{Viscosity solution of HJB}
\subsection{HJB for optimal control problem}
Consider an affine system
\begin{align}
\label{affine system}
		\dot{x} &= a(x)x+b(x)u, \quad u(t)\in U, \quad t\in \left[0,T\right]
\end{align}
with the cost function

\begin{align}
\label{cost_function}
	V(t,x) &= \inf_{u \in U} \int_{0}^{T} \left(x^TQx +u^TRu\right) dt +\psi(x(T))=\inf_{u \in U} \int_{0}^{T} L(x,u) dt +\psi(x(T))
\end{align}
where $x \in R^n$ and $u \in R^m$, and $Q,R$ are positive definite matrix, and the set $U$ is all of the admissible control values as defined below.

\begin{align}
	U:= \{ u: \mathbb R \to \mathbb R^m \ \ measureable, \ \ \ u(t)\in U \ for \ a.e. \ t \}
\end{align}

Given the initial data $x(s)=y$, consider the control u and the corresponding trajectory x, we have the following theorem.

\begin{theorem}
	Given an affine system in eq.\ref{affine system}, and the optimal cost function defined in eq.\ref{cost_function}, then $V^{*}(t,x)$ satisfies the following HJB equation  
\begin{align}
\label{HJB_finite}
		V_t + \inf_{u \in U} \{ L(x,u) + V_x\dot{x}\} = V_t + H(x,u,V_x) = 0 
\end{align}
where the $H(x,u,V_x)$ is usually defined as the Hamiltonian
\end{theorem}

\begin{proof}
First define a function  $\Phi^u(t)$ as 
\begin{align*}
		\Phi^u(t)  := \int_s^tL(x(t),u(t))dt+V^*(t,x),    \quad   t\in \left[ s,T\right]
		\end{align*}
	From the dynamic programming principle
\begin{align}
		 V^*(t_1,x(t_1)) =\inf_{u \in U} \{ \int_{t_1}^{t_2}\left[ L(x(t),u(t))dt+V^*(t_2,x_2) \right]\}
\end{align}	
it is easy to see that 

	\begin{align}
		\Phi^u(t_1) - \Phi^u(t_2) = V^*(t_1,x(t_1)) -\int_{t_1}^{t_2}\left[ L(x(t),u(t))dt+V^*(t_2,x_2) \right]\leq 0
	\end{align}	
	Thus function $\Phi^u(t)$ is non-decreasing, and we have  
\begin{align}
		\frac{d}{d\tau }
		\int_{s}^{\tau}\left[ L(x(t),u(t))dt+V^*(\tau,x_{\tau}) \right] & \geq 0 \\
		V^*_t +  \{ L(x,u) + V_x\dot{x}\}  & \geq 0 
\end{align}		
Notice that the control $u$ is optimal if and only if the cost function is minimum $V^*(s,y)$, hence the function $\Phi^u(t)$ is constant (optimal cost) if $u = u^*$, $\forall t \in [s,T]$
we have
\begin{align}
		\frac{d}{d\tau }
		\int_{s}^{\tau}\left[ L(x(t),u(t))dt+V^*(\tau,x_{\tau}) \right] & = 0 \\
		V^*_t +  \{ L(x,u^*) + V_x(a(x)+b(x)u^*)\}  & = 0 
\end{align}

Through the above equations, it is easy to see that at a point $(s,y)$	where it is differentiable, the value function satisfies the HJB	in eq.\ref{HJB_finite}
\end{proof}
\begin{remark}
	It is worth mentioning that, although in some literatures, e.g. \cite{bressan2011viscosity,cannarsa2004semiconcave}, the HJB is given in the form
	\begin{align}
		-\{V_t + \inf_{u \in U} \{ L(x,u) + V_x\dot{x}\}\} = -V_t + \sup_{u \in U} \{ -V_x\dot{x} -L(x,u)\} = 0 
\end{align}
Although the difference of HJB above with eq.\ref{HJB_finite} is the minus before the whole equation, the corresponding optimal cost function (viscosity solution) are totally different \cite{bressan2011viscosity,bardi1997optimal}. The reason could be see clearly from the definition \ref{viscosity_solution_def1} in next subsection.
%

It is worth mentioning that, the solution of HJB in eq.\ref{HJB_finite} does not only share the meaning of optimal cost in our problem, but also can serve as the Lyapnov function for infinite horizon optimization problem, when the $Q$ and $R$ in the cost function \ref{cost_function} are positive definite and semi positive definite. For example, it is easy to see that in infinite horizon optimization problem,   $V(x)$ is greater or equal than 0 and has the global minimum at $V(0)=0$. Further more,
\begin{align}
	V^*_x\dot{x}+L(x,u^*) =0& \\
	\frac{dV}{dt} = - L(x,u^*)<0& 
\end{align}
\end{remark}


\subsection{Viscosity solution}
\label{sec:viscosity introduce}

\begin{definition}
\label{viscosity_solution_def1}
Given a general nonlinear first order PDE
\begin{align}
\label{HJB}
	F(x, V(x), DV(x)) & =0 & & \text { in } {R}^n \times[0, \infty) \\
	V & =g & & \text { on } {R}^n \times\{t=T\} \nonumber
\end{align}
the solution $V(x,t)$ is called the viscosity solution of eq.\ref{HJB}

\begin{align}
\label{viscosity}
	V(x,t):=\left\{ 
	\left.\lim_{\epsilon \to 0}V^{\epsilon}\right| 	
	V^{\epsilon}_t+H(V^{\epsilon}_x,x) =\epsilon \triangle V^{\epsilon}
	\right\}
\end{align}
where the $\triangle$ is the divergence notation, and $\epsilon$ is a small positive number.	
\end{definition}

Some useful properties of viscosity solution are listed below:
\begin{itemize}
	\item Existence. For a nonlinear first order PDE in eq.\ref{HJB}, the viscosity solution exists.
	\item Uniqueness. For such a problem, there exists at most one viscosity solution.
	\item Stability. The viscosity solution is stable in the following sense \\
	for any disturbance $d>0$, $V^d$ is a viscosity solution of 
	\begin{align}
\label{HJBd}
	V^d_t+H(V_x^d, x) & =0 & & \text { in } {R}^n \times(0, \infty) \\
	V^d & =g^d & & \text { on } {R}^n \times\{t=T\} \nonumber
\end{align}
if parameters of the eq.\ref{HJBd} converge to eq.\ref{HJB}, e.g. $H^d \to H$, $g^d \to g$, and $V^d \to V$ correspondingly, then $V$ is the viscosity solution of eq.\ref{HJB}.
\end{itemize}
The proof of the above properties could be find in lots of books \cite{evans2022partial,bressan2011viscosity,bardi1997optimal}

\begin{remark}
	The name of viscosity solution comes from the right hand side of the equation above: the vanishing viscosity term. With the quasilinear elliptic equations above, for any $\epsilon>0$, eq.\ref{viscosity} has a unique and bounded solution $V^{\epsilon} \in C^2(R^n)$\cite{bardi1997optimal}. When the viscosity term vanishes from some small positive quantity to zero,  $V^{\epsilon}$ converges to the unique solution. Apart from defining the uniqueness of solution, viscosity solution is also a general solution that satisfies equation almost everywhere.
	\end{remark}

It is worth mentioning that, another definition of the viscosity solution which is more commonly used for analysis is given below

\begin{theorem}
	A function $V(t,x) \in C(\Omega)$ is a viscosity solution \cite{bressan2011viscosity} of \ref{HJB} if 
	\begin{align}
	\label{subsolution}
	V(T,x)=g(x) & &\forall x \in R^n  \\
		q+H(p, x,u) & \leq0 & \forall (t,x) \in (0,T)\times \Omega, (q,p) \in D^+V(t,x) \\
q+H(p, x,u) & \geq0 & \forall (t,x) \in (0,T)\times \Omega, (q,p) \in D^-V(t,x) 
			\end{align}
			where $D_{t,x}^{1,+}V(x) $ and $D_{t,x}^{1,-}V(x) $ are the superdifferentials and subdifferentials of V at x which are given
	\begin{align}
	\label{supersolution}
		D_{t,x}^{1,+}V(t,x) =\left\{
		\begin{bmatrix} q\\p \end{bmatrix} 
		\in R\times R^n \left. \right| \varlimsup_{
		\scriptsize {
		\begin{array}{c}   
		    y \to x \\  
    s \to t, s \in (0,T) \\  
  \end{array}}
  }  \frac{V(s,y)-V(t,x)-q(s-t)-<p,y-x>}{\left|y-x\right|+\left|s-t\right|}\leq 0\right\}
	\end{align}
	\begin{align}
	\label{subsolution}
		D_{t,x}^{1,-}V(t,x) =\left\{
		\begin{bmatrix} q\\p \end{bmatrix}
		 \in R\times R^n \left. \right| \varliminf_{
		\scriptsize {
		\begin{array}{c}   
		    y \to x \\  
    s \to t, s \in (0,T) \\  
  \end{array}}
  }  \frac{V(s,y)-V(t,x)-q(s-t)-<p,y-x>}{\left|y-x\right|+\left|s-t\right|}\geq 0\right\}
	\end{align}
Further more, if $s \to t^+$, then we get the right sup and subdifferentials $D_{t^+,x}^{1,+}V(t,x)$ and  $D_{t^+,x}^{1,-}V(t,x)$. Similarly, if $s \to t^-$, then we get the left sub and subdifferentials $D_{t^-,x}^{1,+}V(t,x)$ and  $D_{t^-,x}^{1,-}V(t,x)$
\end{theorem}

\begin{remark}
	From the definition, it could be seen that a vector $\begin{bmatrix} q^T p^T \end{bmatrix}^T$ is a sub differential if and only if the hyperplane $(s,y) \to V(t,x)+q(s-t)+<p,y-x>$ is tangent from below to the curve of V at the point (t,x). The supdifferential is similar but in the opposite direction. It is also worth to mention that the viscosity solution defined by this way is equivalent to the definition \ref{viscosity_solution_def1}
\end{remark}

\section{Uniqueness of semiconvex solution in viscosity sense}
\label{sec:proof of viscostiy}
This section proves the uniqueness of a semi convex solution that satisfy eq.\ref{HJB}, given the assumption that Hamiltonian is concave. 
\subsection{Proof of concave Hamiltonian}
\label{sec:3.1}
\begin{lemma}
\label{concavehamiltonian}
	For affine system in eq.\ref{affine system}
with the cost function in eq.\ref{cost_function},
the hamiltonian is given by 
\begin{align}
\label{Hamiltonian}
		\min_{u \in U} H(V_x, x, u) = 
		\min_{u \in U}\left\{ V_x \dot{x} + x^TQx +u^TRu
	\right\}
\end{align}
when the Lagrangian L is unbounded below in x, the Hamiltonian takes on the value $-\infty$. Since the Hamiltonian is the pointwise minimum of a family of affine function of $V_x$. It is concave, even when problem eq.\ref{cost_function} is not convex\cite{boyd2004convex}.
\end{lemma}

\begin{proof}
	suppose that the Hamiltonian in eq.\ref{Hamiltonian} is convex, there must exists $x_1, x_2, t$ such that
	\begin{align*}
	H(V_{x_1}+t(V_{x_2}-V_{x_1}))<H(V_{x_1})+t(H(V_{x_2})-H(V_{x_1}))
\end{align*}
but this impossible, since the line $H(V_{x_1})+t(H(V_{x_2})-H(V_{x_1}))$ is already the minimum.
\end{proof}
\subsection{Equivalence of convex solution and viscosity solution}
Most reference gives the proof in the form of opposite assumptions and opposite conclusions\cite{bardi1997optimal,cannarsa2004semiconcave}, e.g. when the Hamiltonian is convex, the unique viscosity solution is concave. This part of proof would try to give the opposite theorem that is more familiar to control community and the entire proof refers to the skill in\cite{cannarsa2004semiconcave}.

\begin{definition}
  \label{reachable gradients}
	Let V(t,x): $[0,T]\times R^n \to R$ be locally lipschitz. A vector $[q^T p^T]^T \in R^n$ is called a reachable gradient of V(t,x) at $(t,x) \in [0,T]\times R^n$ if a sequence $\left\{[t_k^T x_k^T]^T\right\} \subset [0,T]\times R^n -  \left\{[t^T x^T]^T\right\}$ exists such that V is differentiable at $[t_k^T x_k^T]^T$ for each $k \in N$, and
	\begin{equation}
		\lim_{k \to \infty} 
		\begin{bmatrix} t_k \\ x_k \end{bmatrix} 
		= 
		\begin{bmatrix} t \\ x \end{bmatrix}
		, \lim_{k \to \infty} D^*V(t_k,x_k) = \begin{bmatrix} q \\ p \end{bmatrix}, 
	\end{equation}
	The set of all reachable gradients of V at (t,x) is denoted by $D^*V(t,x)$
\end{definition}
%


	

\begin{lemma}
\label{lemma:2 supdiff}
	Let V(t,x): $[0,T]\times R^n \to R$, $V(t,\cdot)$ be a convex function. Then the following properties hold true
	
   $D_{t,x}^{1,+}V(t,x)$ are nonempty if and only if V is differentiable at (t,x), so
 		\begin{equation}
			D_{t,x}^{1,+}V(t,x)=D_{t,x}^{1,-}V(t,x)=DV(t,x)
		\end{equation}

%
\end{lemma}
The proof is given in \cite[p9]{bressan2011viscosity} and \cite[p51,p58]{cannarsa2004semiconcave}

\begin{lemma}
\label{convexhull}
	Let V(t,x): $[0,T]\times R^n \to R$, $V(t,\cdot)$ be a convex function and let $x \in A$. Then
	\begin{equation}
		D_{t,x}^{1,-}V(t,x)=CH[D^*V(t,x)]
	\end{equation}
	where CH represents the convex hull notation.
\end{lemma}
\begin{proof}
	The proof could be referred to \cite[59]{cannarsa2004semiconcave} when taking into $-V(x)$ as semiconvex
\end{proof}
\begin{remark}
	it is easy to see that $D^*V(x)$ is a compact set, and it is bounded since we are takeing $V(t,\cdot)$ Lipschitz. These propositions shows that, not only $D^*V(x) \subset  D^-V(x)$ but also that all reachable gradients are boundary points of $D^-V(x)$
	\end{remark}
\begin{theorem}
\label{theorem:finite}
	Given a system in eq.\ref{affine system}, considering the time limited cost function in eq.\ref{cost_function} and the HJB equation in eq.\ref{HJB_finite}. $V(t,x)$ is the viscosity solution of the above HJB, if
	\begin{itemize}
	  \item  $V(t,x)$ satisfies the eq.\ref{HJB_finite} almost everywhere
		\item $V(t,\cdot)$ is locally liptschitz convex and $H(x,u,\cdot)$ is locally liptschitz concave
		\item $V(\cdot,x) \in C^1$ 
	\end{itemize}
\end{theorem}
\begin{proof}
	First we observe that for $\forall (t,x) \in (0,T)\times R^n$
	\begin{equation}
 q + H(x,u,p) = 0,  \quad  \forall [q^T,p^T]^T \in D^*V(t,x)
	\end{equation}
	 This follows directly from the definition of $D^*V(x)$, the conditions $V(\cdot,x) \in C^1$ and continuity of H. Since in Lemma \ref{convexhull}, $D^-V(t,x)$ is the convex hull of $D^*V(t,x)$, and $H(x,u,\cdot)$ is concave guaranteed from lemma \ref{concavehamiltonian} , it is easy to see that 
	\begin{align}
q + H(p, x,u) & \geq0 & \forall (t,x) \in (0,T) \times R^n, [q^T p^T]^T \in D^-V(t,x)
	\end{align}
The next step is to check the 
	\begin{align}
		q + H(p, x,u) & \leq0 & \forall (t,x) \in (0,T) \times R^n, [q^T p^T]^T \in D^+V(t,x) 
	\end{align}
If there exists some $[q^T p^T]^T \in D^+V(t,x)$, according to the lemma \ref{lemma:2 supdiff}, V(t,x) should be differentiable and $[q^T p^T]^T = D^V(t,x)$. Thus the inequality holds as an equality. Hence $V(t,x)$ is the viscosity solution for eq.\ref{HJB}, and uniqueness is ensured.
\end{proof}
\begin{corollary}
	Given a system in eq.\ref{affine system} and consider infinite horizon optimal control problem:
\begin{align}
	V(x) &= \inf_{u \in U} \int_{0}^{\infty} \left(x^TQx +u^TRu\right) dt 
\end{align}
where the HJB equation becomes 
\begin{align}
\label{infinite_HJB}
		\min_{u \in U} H(V_x, x, u) = 
		\min_{u \in U}\left\{ V_x \dot{x} + x^TQx +u^TRu
	\right\}=0
\end{align}

	Let $V(x)$ be a locally liptschitz convex function satisfying eq.\ref{infinite_HJB} almost everywhere. If $H(\cdot,x,u)$ is concave, then V(x) is the viscosity solution of eq.\ref{HJB}, thus $V(x)$ is the unique optimal cost function for eq.\ref{infinite_HJB}
\end{corollary}

\begin{proof}
	 We consider the $D^*V(x)$ as the special case of $D^*V(t,x)$, and $D^-V(x)$ is the special case of $D_{t,x}^{1,-}V(t,x)$. It could be seen directly from the definition of $D^*V(x)$ and the continuity of H, that the HJB follows when $\forall p \in D^*V(x)$. Since $D^-V(x)$ is the convex hull of $D^*V(x)$, and $H(\cdot,x,u)$ is concave, it is easy to see that 
	\begin{align}
H(p, x,u) & \geq0 & \forall x \in R^n, p \in D^-V(x) 
	\end{align}
For the super differentials  
	\begin{align}
		H(p, x,u) & \leq0 & \forall x \in R^n, p \in D^+V(x) 
	\end{align}
If there exists some $p \in D^+V(x)$, similarly like theorem \ref{theorem:finite}, $p = DV(x)$. Thus inequality\ref{subsolution} hold as an equality. Hence $V(x)$ is the viscosity solution for eq.\ref{HJB}.\end{proof}

\section{PINN-based numerical viscosity solution of HJB}
\label{sec: numerical method}
This section introduced two methods to ensure the convex structure of solution based on PINN method

\subsection{Method 1: Addding convex term in loss function}
The sufficient and necessary condition of the function $V(t,.)$ being convex is
\begin{align}
	Hess(V_x)=\begin{bmatrix}
	\frac{\partial^2V(t,x)}{\partial x_1^2} & \cdots & \frac{\partial^2V(t,x)}{\partial x_1x_n}\\
	\vdots & \ddots & \vdots\\
	\frac{\partial^2V(t,x)}{\partial x_nx_1} & \cdots & \frac{\partial^2V(t,x)}{\partial x_n^2}
	\end{bmatrix}>0 \\
	\iff \left|Hess(V_x)_{i,i}\right|>0
\end{align}
where $\left|Hess(V_x)_{i,i}\right|$ is the ith order principal minor determinant.

Hence, the following loss function is thus proposed during training
\begin{align}
	Loss = \frac{1}{N_i} \Sigma_{j=1}^{N_i}\left|
	V_{tj}+H(x_j,V_{xj})
	\right|^2
	+\frac{1}{N_b} \Sigma_{j=1}^{N_b}
	\left|V(tj,xj)-V_{NN}(tj,xj)\right|^2 + 
	\frac{1}{N_H} \Sigma_{j=1}^{N_H}\left|\min (0,\left|Hess(V_x)_{j,j}
	\right|)\right|^2
\end{align}
where $N_i$ is the number of inner point, $N_b$ is the number of boundary point, $N_H$ is the number of point calculating Hessin matrix.
It is worth mentioning that the Tanh activation function would be used in the neural network, hence the smoothness of the solution is guaranteed. Once we can check the positiveness of Hessian matrix of $V_x$, then the uniqueness of solution is satisfied according to theorem \ref{theorem:finite}.

\subsection{Method 2: Design of input convex neural network}
Another method is to construct a partially input convex neural network directly to ensure the convexity of solution in theorem\ref{theorem:finite}.

To build a first full input convex neural network, it is necessary to review some basic properties of convex functions

\begin{proposition}
\label{proposition }
	If $f(x)$ and $g(x)$ are both convex functions
		\begin{itemize}
	
		\item  $m(x)=\max \left\{ f(x), g(x)\right\} $ and $h(x)=w_1f(x)+w_2g(x)$ are both convex functions $\forall w_1>0,w_2>0$
		\item Affine transform doesn't change convexity, e.g. $h(p)=f(Ap+b)$ is also convex function
		\item $f(g(x))$ is also convex given that $g(x)$ is monotonically increasing
     \end{itemize}
\end{proposition}

From the proposition above, it is easy to prove that the following neural network is a convex function

\begin{align}
	y = W_1\sigma(W_0x+b_0) + b_1\\
	where \left\{\left. w_i>0\right|\forall w_i \in W_1 \nonumber
\right\},\sigma(x)=\max(0,x) \nonumber
\end{align}
It is clear that this single neural network represents the sum of a bunch of convex functions. To improve the approximation ability, we add a feed forward channel to the output, so that the single layer NN become
\begin{align}
\label{ConvexNN}
	y = W_1\sigma(W_0x+b_0) + b_1 + fx\\
	where\left\{\left. w_i>0\right| \forall w_i \in W_1 \nonumber
\right\}\end{align}
The structure is shown in Fig.\ref{fig:NN} and the following theorem holds

\begin{figure}[!h]
\centering
\label{fig:NN}
\includegraphics[width=10cm, height = 8cm]{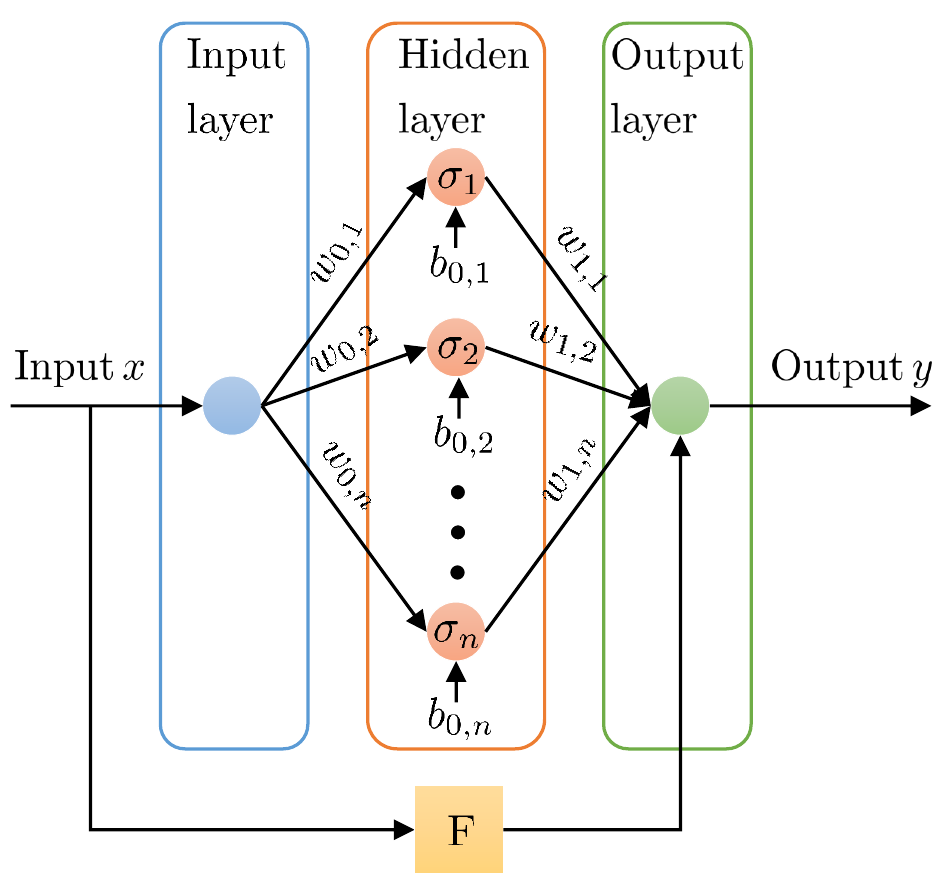}
\caption{Convex Neural Network}
\end{figure}

\begin{theorem}
	Given any convex function $f(x)$, and $\forall C>0$,$\exists y_{NN}(x)$ in the form eq.\ref{ConvexNN}, such that
	\begin{align}
	\label{approxmationtheorem}
		\| f(x)-y_{NN}(x) \|<C
	\end{align}
	where $\| f(x)-y_{NN}(x) \|$ is the $L_\infty$ norm.
\end{theorem}
	
\begin{proof}
	the eq.\ref{ConvexNN} could be re-write in the following form
	\begin{align}
	\label{NNOneChannel}
		y_{NN}(x)= \Sigma_{i=1}^{n} w_{1,i} \{ \max \{ w_{0,i}x +b_{0,i}, \ 0\}+\frac{b_{1,i}}{w_{1,i}}  +\frac{f_{i}}{w_{1,i}}x  \}
	\end{align}
	where i represents the index of ith neural in the hidden layer, $w_{0,i}, w_{1,i}$ are the weight vector corresponds to the relative hidden neuron, $b_{0,i}$ is the corresponding bias. Further more $b_{1,i}, f_{i}$ are the parameters that satisfy
	\begin{align}
		\Sigma_{i=1}^{n} b_{1,i} = b_1, \Sigma_{i=1}^{n} f_{i} = f \nonumber
	\end{align}
	it is easy to see that eq.\ref{NNOneChannel} represents the effect of one channel through a single hidden neuron, and it is equivalent to the following summation of convex two piece linear function
	\begin{align}
		\label{NNOneChannel}
		y_{NN}(x)&= \Sigma_{i=1}^{n} w_{1,i} \{ \max \{ \left(w_{0,i} + \frac{f_{i}}{w_{1,i}}\right)x +b_{0,i} +\frac{b_{1,i}}{w_{1,i}}, \ 
		\frac{f_{i}}{w_{1,i}}x + 
		\frac{b_{1,i}}{w_{1,i}}\}\}	\\
		\label{TwoPieceLinear}
		 &= \Sigma_{i=1}^{n}   \max \{ m_1x +n_{1}, \ 
		m_2x +n_{2}\}
		\end{align}
		it is obvious that any convex piece-wise linear function $h(x)$ could represented by the sum of two convex piece linear function, e.g. eq.\ref{TwoPieceLinear}. A constructive proof could start from the outermost two piece linear function, plus the following function 
		\begin{align}
h_j(x)= \max \{ m_jx +n_{j}, 0\}
		\end{align} 
		when comes to the nonsmooth point. Hence $y_{NN}(x)$ could represents any piecewise linear function. It is also known that convex piecewise linear function is dense in continuous convex function space\cite{cox1971algorithm,gavrilovic1975optimal}, hence eq.\ref{approxmationtheorem} is prooved.
		\end{proof} 


\begin{remark}
	it is worth to point out that, the reason that we don't use the maximum of a bunch of linear functions directly is that, when this network is used for solving PDEs, which involves lots of computation of derivatives, it would cause oscillation during training, and it is almost impossible to converge to the global optimal. On the contrary, when approximating a convex function directly,  this maximum of a bunch linear functions shows good performance.
\end{remark}

Next phase is to build the partially convex network. A partial convex neural network is proposed as Fig.\ref{fig:partialNN} shows, where the input output relation is given as

\begin{align}
	y &= W_2\sigma(W_0x + W_1{\mathcal F}_{{NN}_2}(\hat y) + b_0) + b_1\\
	\hat y &= {\mathcal F}_{{NN}_1}(t) 
\end{align}
where $\left\{\left. w_i>0\right|\forall w_i \in W_2 \nonumber
\right\},\sigma(x)=\max(0,x)$, and ${\mathcal F}_{{NN}_1},{\mathcal F}_{{NN}_2}$ are normal fully connected networks.

On the one hand, for a given $t$, $V(t,.)$ could be seen as a convex neural network in Fig.\ref{fig:NN} with a special bias decided by input t, hence $V(t,.)$ is a convex neural network. On the other hand, for a given $x$, $W_0x$ could be regarded as a special bias for the network of $V(t,.)$. It is worth mentioning that, although $W_2$ should be element wise positive, there exists a network ${\mathcal F}_{{NN}_2}$ with certain layers, so that the input $\hat y$ approximate output $y$. Hence given a normal fully connected network ${\mathcal F}_{{NN}_1}$ with certain layers, the representability from input $t$ to output $y$ is not worth than a normal fully connected neural network. 

\begin{figure}[!h]
\centering
\label{fig:partialNN}
\includegraphics[width = 15cm]{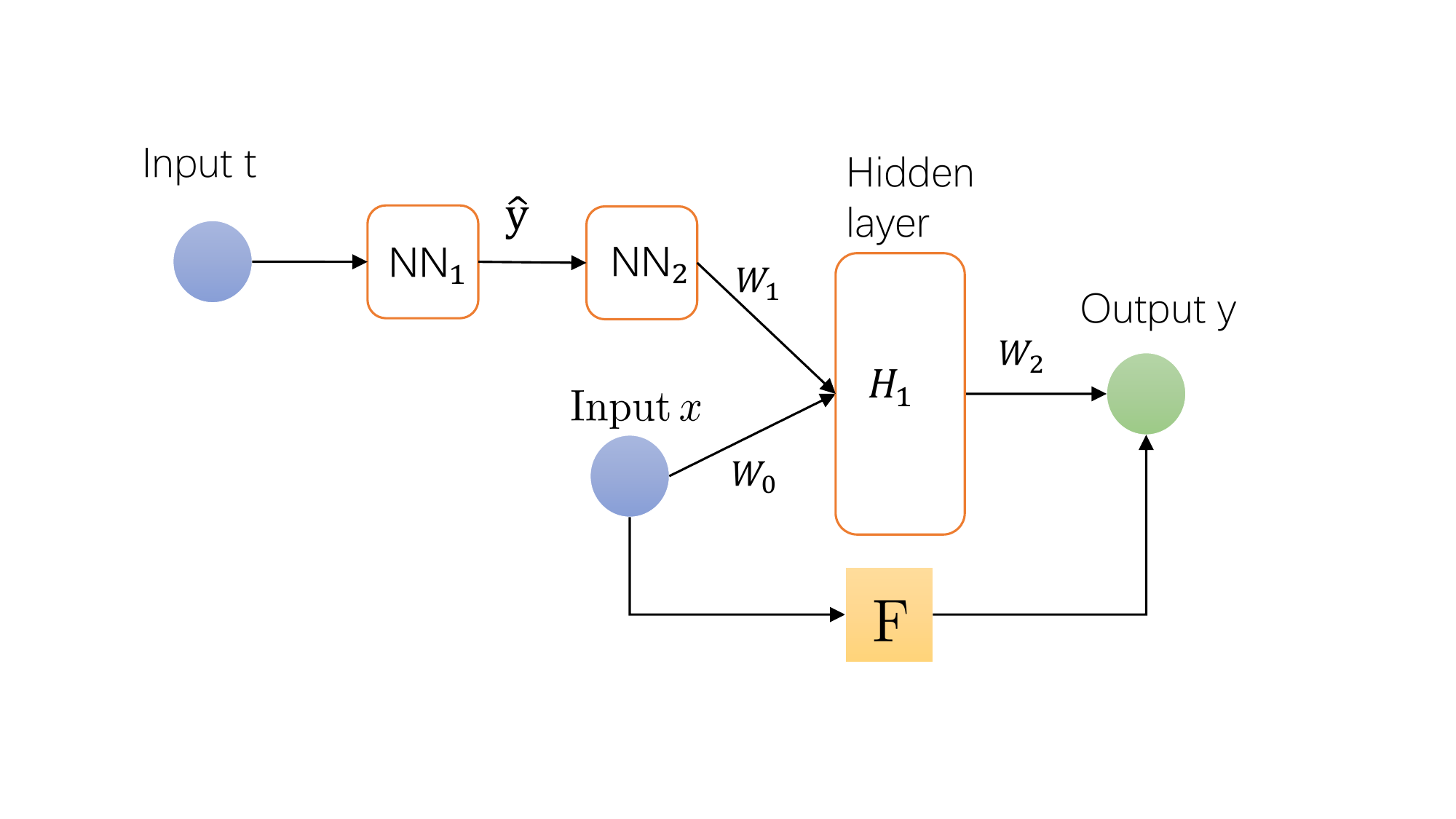}
\caption{Partial Convex Neural Network}
\end{figure}

\subsection{Strip training method}
A key issue during the training of PINN, is how to find global minimum. For HJB-based problem, we found a efficient training method that could avoid local minimum theoritically, we call it strip training method, which is similar to the curriculum learning in reinforcement learning.

For a given system
\begin{align}
		\dot{x} &= a(x)x+b(x)u, \quad u(t)\in U, \quad t\in \left[0,T\right]
\end{align}
with the cost function
\begin{align}
	V(t,x) &= \inf_{u \in U} \int_{0}^{T} \left(x^TQx +u^TRu\right) dt +\psi(x(T))=\inf_{u \in U} \int_{0}^{T} L(x,u) dt +\psi(x(T))
\end{align}
where the cost function $V(T,x)$, at end time T, is known. The corresponding HJB is given as

\begin{align}
		V_t + \inf_{u \in U} \{ L(x,u) + V_x\dot{x}\} = V_t + H(x,u,V_x) = 0 
\end{align}

We have the strip training algorithm which is shown in algorithm \ref{algorithm: strip training}. It worth mentioning that, the method 2 has a better performance than method 1 with less computation time in practice, so this strip training method is studied based on method 2. Although method 1 could also be realized in similar way. The sketch figure is shown in Fig.\ref{fig:training_sketch}
\begin{figure}[!h]
\centering
\label{fig:training_sketch}
\includegraphics[width = 15cm]{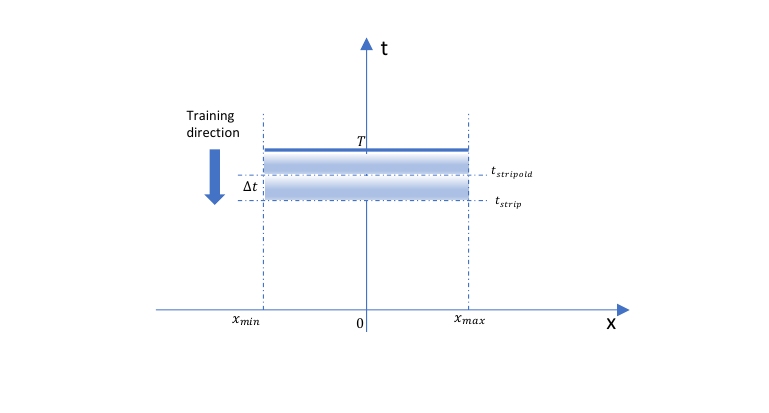}
\caption{Strip training procedure}
\end{figure}

The main idea of strip training is, the optimal function $V(t,x)$ is continuous, and if we start from the boundary, and take $\Delta t$ small enough, the shape of neural network could approximate the true solution as close as we want. Then the training step is first train the boundary point, then under the shape of boundary points (realized by large $\lambda$), add the information of HJB equation to the loss (in step 2.2). Then we treat the trained area as a known area, so after adding a new strip to the training area (in step 3), the large $\lambda$ is used to maintain the shape of known area. Repeat the procedure till the time reaches to 0.

\begin{algorithm}[!htbp]
\renewcommand{\algorithmicensure}{\textbf{Step}}
\caption{Strip training method}
\label{algorithm: strip training}

\begin{algorithmic}
\ENSURE \textbf{1:} Boundary Training
\STATE \textbf{step 1.1}: Set sample region $(t,x) \ in [0,T]\times [x_{min}, x_{max}]$ and sample number $N$
\STATE \textbf{step 1.2}: Suppose that $V(t_i,x_i) = V(T,x_i), i = 0,1,\cdots,N$
\STATE \textbf{step 1.3}:
\WHILE{$\| NN_{PINN}(t_i,x_i) - V(t_i,x_i)\|$ $<$ $l_b$} 
	\STATE  Train $NN$ with Adam optimizer   
\ENDWHILE
\end{algorithmic}

\begin{algorithmic}
\ENSURE \textbf{2:} Boundary with initial strip training
\STATE \textbf{step 2.1}: Set sample region $(t,x) \ in [t_{ini},T]\times [x_{min}, x_{max}]$ and sample number $N$
\STATE \textbf{step 2.2}:  Set the boundary points $V(t_{ib},x_{ib}) = V(T_{ib},x_{ib})$, and inner points $V(t_{in},x_{in})$, where $(t_{in},x_{in}) \in [t_{ini},T]\times [x_{min}, x_{max}]$, and  $loss = \lambda\| NN(t_{ib},x_{ib}) - V(t_{ib},x_{ib})\| + \|\frac{\partial NN(t,x)}{\partial t} + H(x,u,\frac{\partial NN(t,x)}{\partial x}) \|$ 
\STATE \textbf{step 2.3}:
\WHILE{ $ loss < loss_{th1}$} 
	\STATE  Train $NN(t,x)$ with a large $\lambda$ and Adam optimizer.	
\ENDWHILE

\STATE \textbf{step 2.4}:
\WHILE{ $ loss < loss_{th2}$} 
	\STATE Set $\lambda = 1$. Train $NN(t,x)$ Adam optimizer.
	\ENDWHILE
\end{algorithmic}

\begin{algorithmic}[!H]
\ENSURE \textbf{3:} Strip training
\STATE \textbf{step 3.1}: Set  $t_{stripold} = t_{ini}$ and $t_{strip} = t_{stripold} - \Delta t$.
 The sample region is $(t,x) \in [t_{strip},T]\times [x_{min}, x_{max}]$ and sample number $N$.
 \STATE \textbf{step 3.2}:  Set the boundary points $V(t_{ib},x_{ib}) = V(T_{ib},x_{ib})$, and inner points $V(t_{iold},x_{iold})$, where $(t_{iold},x_{iold}) \in [t_{stripold},T]\times [x_{min}, x_{max}]$, and 
$V(t_{inew},x_{inew})$, where $(t_{inew},x_{inew}) \in [t_{strip},t_{stripold}]\times [x_{min}, x_{max}]$.
 Training loss is set to be  $loss = \lambda\{
 \| NN(t_{ib},x_{ib}) - V(t_{ib},x_{ib})\| + \|\frac{\partial NN(t,x)}{\partial t}|_{(t,x)=(t_{iold},x_{iold})}
  + H(x,u,\frac{\partial NN(t,x)}{\partial x})|_{(t,x)=(t_{iold},x_{iold})} \|\}
  + \|\frac{\partial NN(t,x)}{\partial t}|_{(t,x)=(t_{inew},x_{inew})}
  + H(x,u,\frac{\partial NN(t,x)}{\partial x})|_{(t,x)=(t_{inew},x_{inew})} \|
 $  
\STATE \textbf{step 3.3}:
\WHILE{ $ loss < loss_{th1}$} 
	\STATE  Train $NN(t,x)$ with a large $\lambda$ and Adam optimizer.	
\ENDWHILE

\STATE \textbf{step 3.4}:
\WHILE{ $ loss < loss_{th2}$} 
	\STATE Set $\lambda = 1$. Train $NN(t,x)$ Adam optimizer.
	\ENDWHILE

\STATE \textbf{step 3.5}:
\IF{ $tstrip > 0$} 
	\STATE $t_{stripold} = t_{strip}$, $t_{strip} = t_{stripold} - \Delta t$, go to step 3.2
\ELSE
	\STATE Stop training
\ENDIF
	
\end{algorithmic}

\end{algorithm}

\section{Simulation results}

\label{sec:simulation}
Here we provide three nonlinear examples.
\subsection{Example 1}

	\begin{equation}
	V^{*} = \min_{u \in U} \int_{0}^{\infty} x^Tx +u^2 =0.5x_1^2+x_2^2
\end{equation}
where
\begin{align*}
		\dot{x_1} &= -x_1+x_2 \\
		\dot{x_2} &= -0.5x_1-0.5x_2+0.5x_1^2x_2 +x_1u \\
\end{align*}
The corresponding HJB is 
\begin{align}
		x_1^2+x_2^2+(-x_1+x_2)V_{x_1}+(-0.5x_1-0.5x_2+0.5x_1^2x_2)V_{x_2} -0.25x_1^2V_{x_2}^2=0\\
	V(0)=0
\end{align}

The training data is sampled with the range $[-1,-1]$ to $[1,1]$.
results with both methods are shown in Fig.\ref{fig:Ex1}. The orange surface is the Analytical solution, while the blue surface is the prediction of neural network. It could be seen that in this case, the method 2 has a better accuracy.

\begin{figure}[!h]
\label{fig:Ex1}
\centering
\includegraphics[width=18cm]{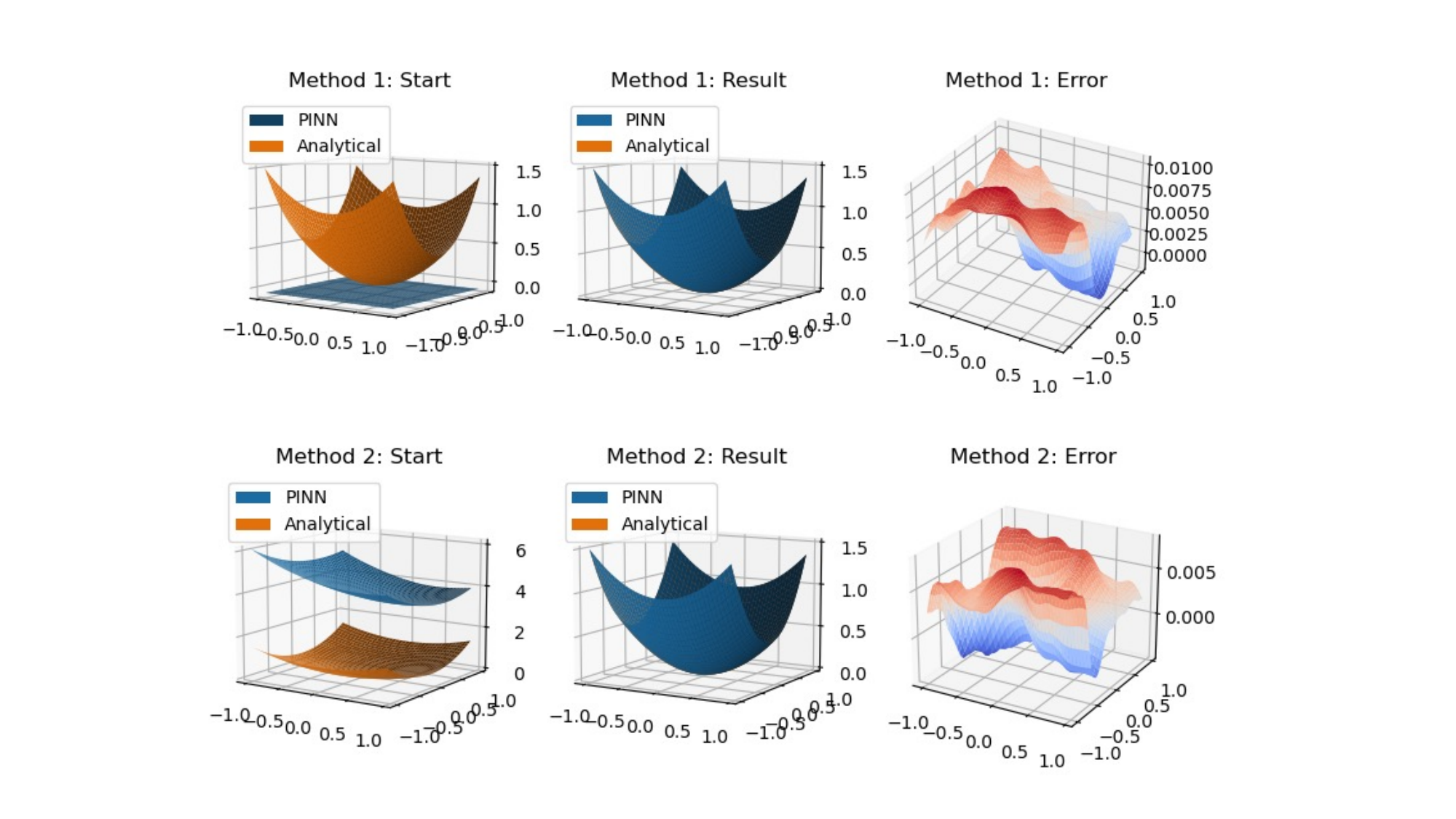}
\caption{Example 1 by method 1}
\end{figure}




\subsection{Example 2}
	\begin{equation}
	V^{*} = \min_{u \in U} \int_{0}^{\infty} x^Tx +u^2 =0.5x_1^2+x_2^2
\end{equation}
where
\begin{align*}
		\dot{x_1} &= -x_1+x_2 \\
		\dot{x_2} &= -0.5x_1-0.5x_2(1-(cos(2x_1)+2)^2)+(cos(2x_1)+2)u \\
\end{align*}
The corresponding HJB is 
\begin{align}
		x_1^2+x_2^2+(-x_1+x_2)V_{x_1}+(-0.5x_1-0.5x_2(1-(cos(2x_1)+2)^2))V_{x_2} -0.25(cos(2x_1)+2)^2V_{x_2}^2=0\\
	V(0)=0
\end{align}

The training data is sampled with the range $[-1,-1]$ to $[1,1]$.
results with both methods are shown in Fig.\ref{fig:Ex2} . The orange surface is the Analytical solution, while the blue surface is the prediction of neural network. In this case, because of the nonlinearity of system and the convex optimization term, the method 1 could not find the global optimum. While the method 2 find the solution directly.

\begin{figure}[!h]
\label{fig:Ex2}
\centering
\includegraphics[width=18cm]{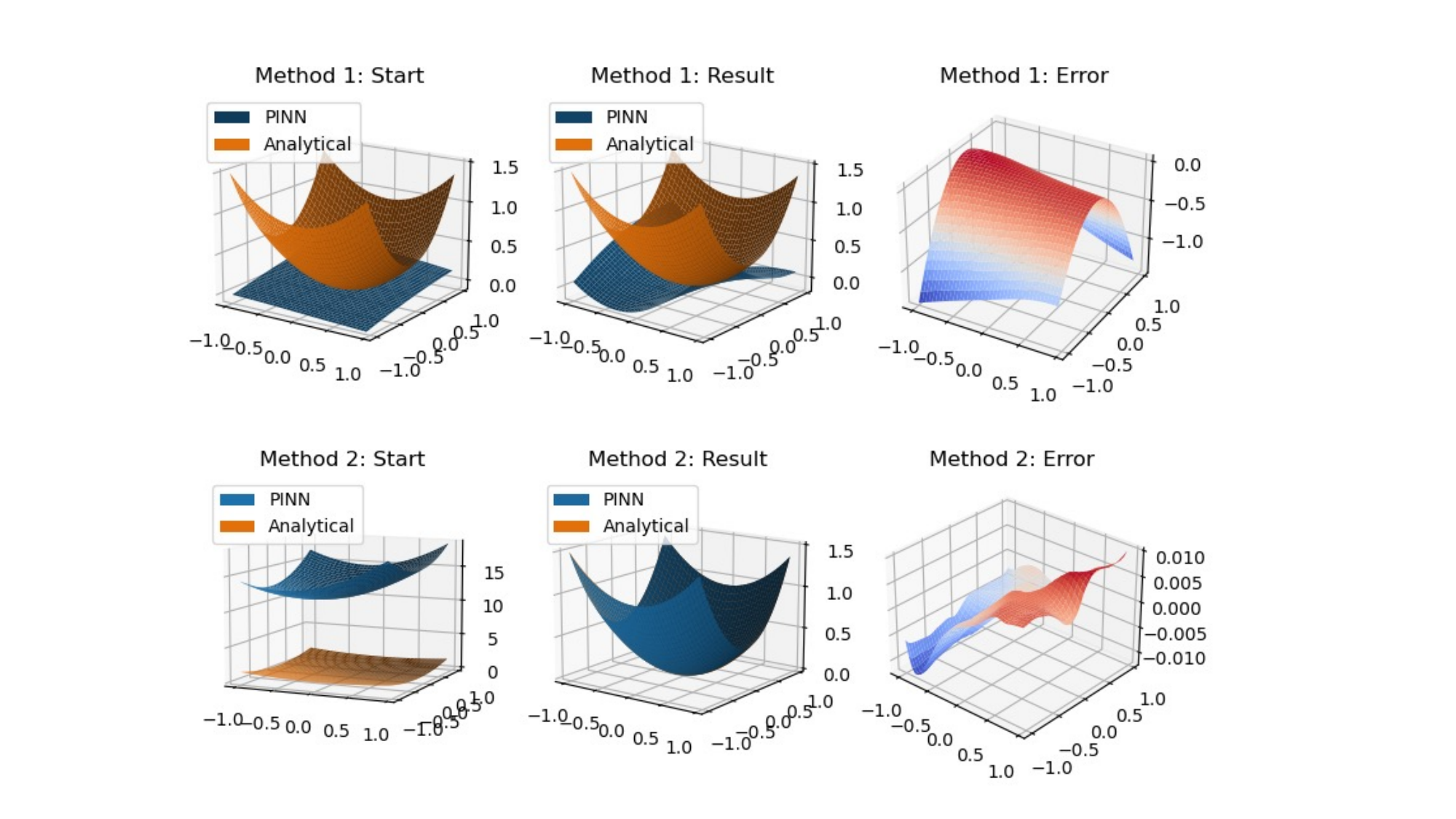}
\caption{Example 2 by method 1}
\end{figure}

%

\subsection{Example 3}

	\begin{equation}
	V^{*} = \min_{u \in U} \int_{0}^{\infty} x_2^2+u^2 =x_1^2(\frac{\pi}{2}+arctan(5x_1))+x_2^2
\end{equation}
where
\begin{align*}
		\dot{x_1} &= x_2 \\
		\dot{x_2} &= -x_1(\frac{\pi}{2}+arctan(5x_1))-\frac{5x_1^2}{2+50x_1^2}+4x_2 +3u \\
\end{align*}
The corresponding HJB is 
\begin{align}
		x_2^2+x_2V_{x_1}+(-x_1(\frac{\pi}{2}+arctan(5x_1))-\frac{5x_1^2}{2+50x_1^2}+4x_2)V_{x_2} -\frac{9}{4} V_{x_2}^2=0\\
	V(0)=0
\end{align}
This example is the most hard one to be trained. The results are given in Fig.\ref{fig:Ex3}. To show that we don't start at some point that is already close to the true solution, the initial function has a really large distance from the true solution, and by method 2, we could see that the PINN could convergent to the true solution. But the accuracy is not as good as the two example before. 

\begin{figure}[!h]
\label{fig:Ex3}
\centering
\includegraphics[width=18cm]{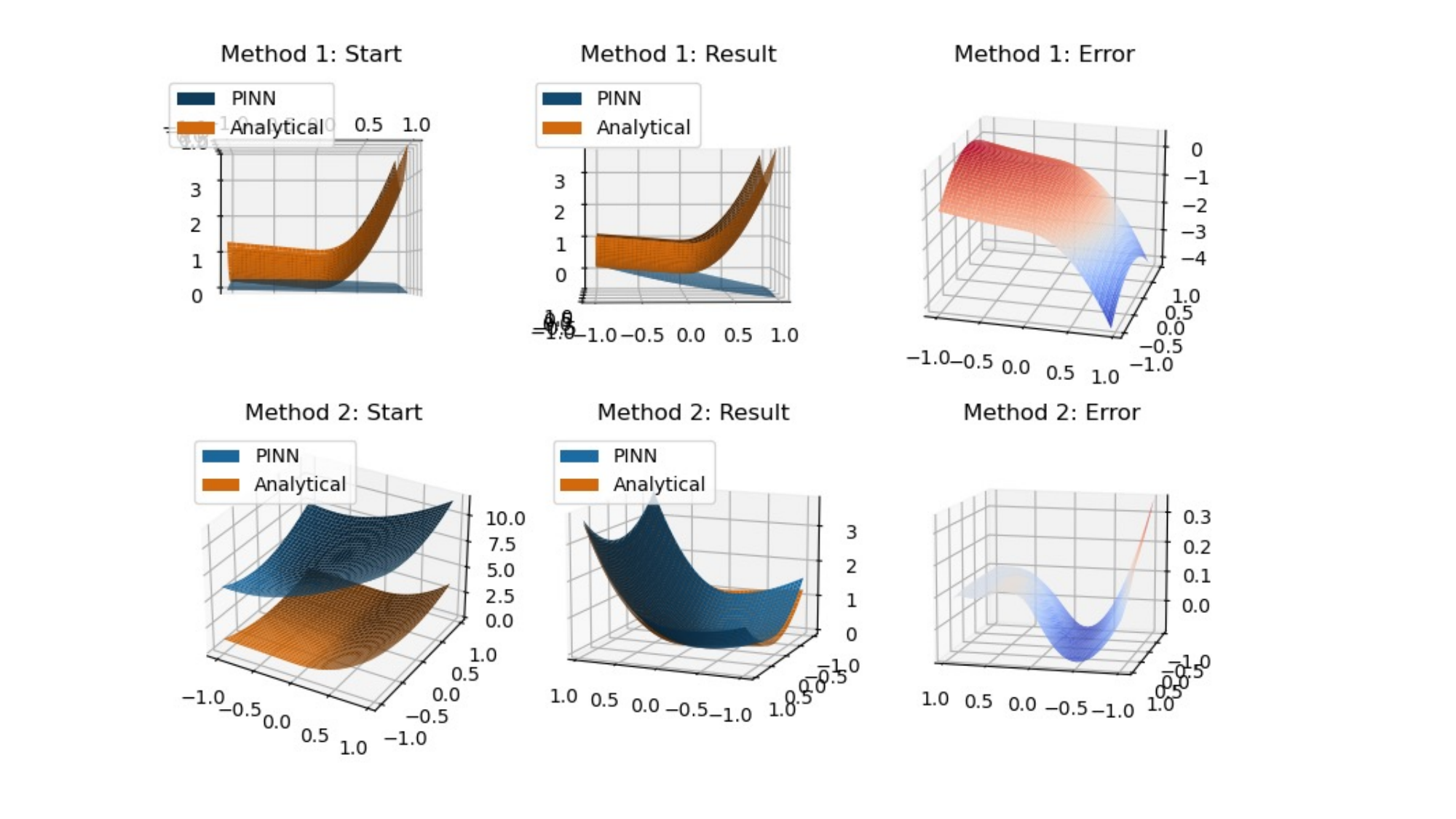}
\caption{Example 3 by method 1}
\end{figure}


\subsection{Example 4}
For a linear system
\begin{align*}
		\dot{x}(t) &= x(t) + u(t) 
\end{align*}
the time limited quadratic cost function is given as
\begin{equation}
	V^{*}(t,x) = \min_{u \in U}V(t,x)=  x^2(T) + \int_{t}^{T} u^2 dt=\frac{2x^2}{1+ e^{2t-2T}} 
\end{equation}
where $T= 10$ in this example. and corresponding HJB is 
\begin{align}
	V_t - 0.25V^2_x + V_xx = 0 \\
	V(10,x)=x^2(10)
\end{align}

The corresponding training result is shown in Fig.\ref{fig:Ex4}. It could be seen that, for time limited case, the method 2 could well approximate the true solution of the HJB.
\begin{figure}[!h]
\label{fig:Ex4}
\centering
\includegraphics[width=17cm]{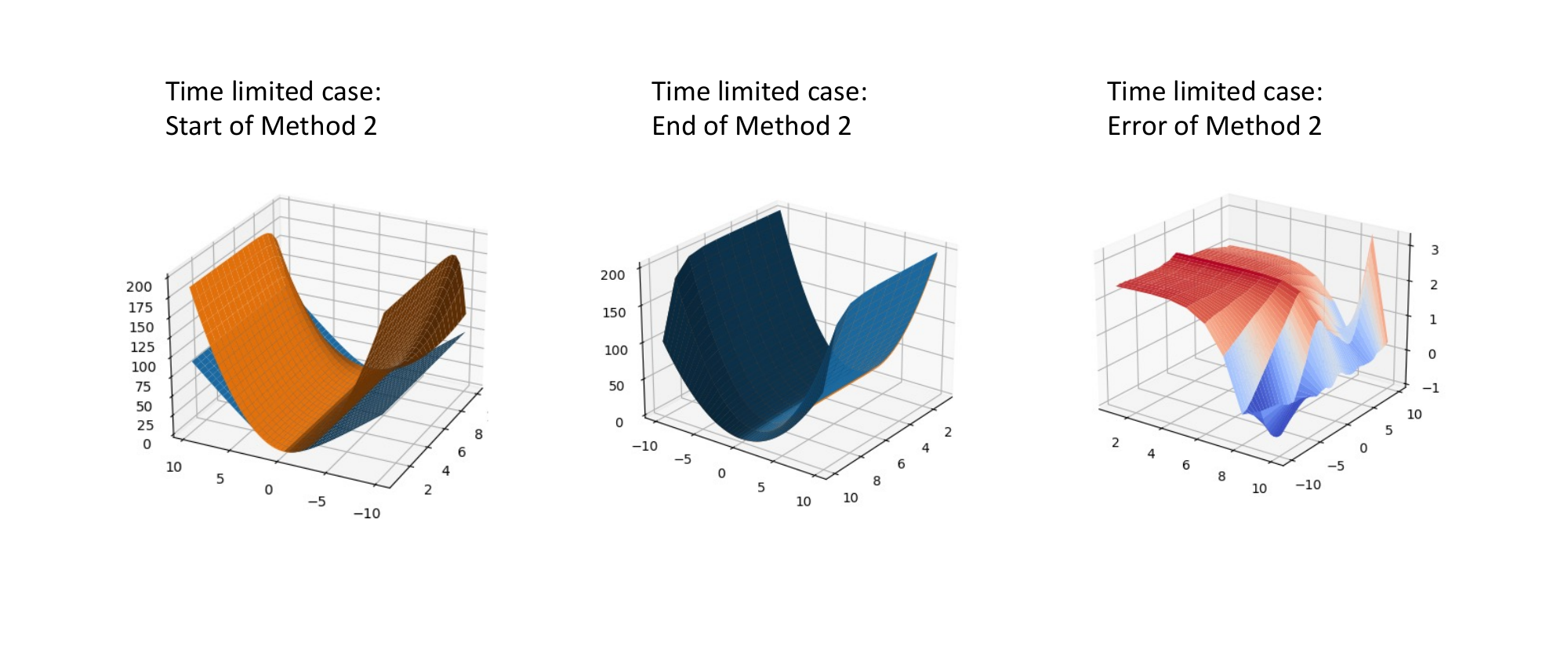}
\caption{Example 4 by method 2}
\end{figure}

\bibliographystyle{IEEEtranS}
\bibliography{sample.bib}

\begin{thebibliography}{10}
\providecommand{\url}[1]{#1}
\csname url@samestyle\endcsname
\providecommand{\newblock}{\relax}
\providecommand{\bibinfo}[2]{#2}
\providecommand{\BIBentrySTDinterwordspacing}{\spaceskip=0pt\relax}
\providecommand{\BIBentryALTinterwordstretchfactor}{4}
\providecommand{\BIBentryALTinterwordspacing}{\spaceskip=\fontdimen2\font plus
\BIBentryALTinterwordstretchfactor\fontdimen3\font minus \fontdimen4\font\relax}
\providecommand{\BIBforeignlanguage}[2]{{%
\expandafter\ifx\csname l@#1\endcsname\relax
\typeout{** WARNING: IEEEtranS.bst: No hyphenation pattern has been}%
\typeout{** loaded for the language `#1'. Using the pattern for}%
\typeout{** the default language instead.}%
\else
\language=\csname l@#1\endcsname
\fi
#2}}
\providecommand{\BIBdecl}{\relax}
\BIBdecl

\bibitem{bardi1997optimal}
M.~Bardi, I.~C. Dolcetta \emph{et~al.}, \emph{Optimal control and viscosity solutions of Hamilton-Jacobi-Bellman equations}.\hskip 1em plus 0.5em minus 0.4em\relax Springer, 1997, vol.~12.

\bibitem{boyd2004convex}
S.~P. Boyd and L.~Vandenberghe, \emph{Convex optimization}.\hskip 1em plus 0.5em minus 0.4em\relax Cambridge university press, 2004.

\bibitem{bressan2011viscosity}
A.~Bressan, ``Viscosity solutions of hamilton-jacobi equations and optimal control problems,'' \emph{Lecture notes}, 2011.

\bibitem{cai2021physics}
S.~Cai, Z.~Mao, Z.~Wang, M.~Yin, and G.~E. Karniadakis, ``Physics-informed neural networks (pinns) for fluid mechanics: A review,'' \emph{Acta Mechanica Sinica}, vol.~37, no.~12, pp. 1727--1738, 2021.

\bibitem{cannarsa2004semiconcave}
P.~Cannarsa and C.~Sinestrari, \emph{Semiconcave functions, Hamilton-Jacobi equations, and optimal control}.\hskip 1em plus 0.5em minus 0.4em\relax Springer Science \& Business Media, 2004, vol.~58.

\bibitem{cox1971algorithm}
M.~G. Cox, ``An algorithm for approximating convex functions by means by first degree splines,'' \emph{The Computer Journal}, vol.~14, no.~3, pp. 272--275, 1971.

\bibitem{evans2022partial}
L.~C. Evans, \emph{Partial differential equations}.\hskip 1em plus 0.5em minus 0.4em\relax American Mathematical Society, 2022, vol.~19.

\bibitem{furfaro2022physics}
R.~Furfaro, A.~D'Ambrosio, E.~Schiassi, and A.~Scorsoglio, ``Physics-informed neural networks for closed-loop guidance and control in aerospace systems,'' in \emph{AIAA SCITECH 2022 Forum}, 2022, p. 0361.

\bibitem{gavrilovic1975optimal}
M.~M. Gavrilovi{\'c}, ``Optimal approximation of convex curves by functions which are piecewise linear,'' \emph{Journal of Mathematical Analysis and Applications}, vol.~52, no.~2, pp. 260--282, 1975.

\bibitem{he2023learning}
D.~He, S.~Li, W.~Shi, X.~Gao, J.~Zhang, J.~Bian, L.~Wang, and T.-Y. Liu, ``Learning physics-informed neural networks without stacked back-propagation,'' in \emph{International Conference on Artificial Intelligence and Statistics}.\hskip 1em plus 0.5em minus 0.4em\relax PMLR, 2023, pp. 3034--3047.

\bibitem{mukherjee2023bridging}
A.~Mukherjee and J.~Liu, ``Bridging physics-informed neural networks with reinforcement learning: Hamilton-jacobi-bellman proximal policy optimization (hjbppo),'' \emph{arXiv preprint arXiv:2302.00237}, 2023.

\bibitem{raissi2019physics}
M.~Raissi, P.~Perdikaris, and G.~E. Karniadakis, ``Physics-informed neural networks: A deep learning framework for solving forward and inverse problems involving nonlinear partial differential equations,'' \emph{Journal of Computational physics}, vol. 378, pp. 686--707, 2019.

\bibitem{rust1997using}
J.~Rust, ``Using randomization to break the curse of dimensionality,'' \emph{Econometrica: Journal of the Econometric Society}, pp. 487--516, 1997.

\bibitem{sahli2020physics}
F.~Sahli~Costabal, Y.~Yang, P.~Perdikaris, D.~E. Hurtado, and E.~Kuhl, ``Physics-informed neural networks for cardiac activation mapping,'' \emph{Frontiers in Physics}, vol.~8, p.~42, 2020.

\end{thebibliography}

\end{document}